\documentclass[letterpaper, 10pt, conference]{ieeeconf}

\IEEEoverridecommandlockouts                              % This command is only needed if 
\usepackage{cite}
\let\labelindent\relax
\usepackage{enumitem}
\usepackage{soul}
\usepackage{url}
\usepackage[hidelinks]{hyperref}
\usepackage[utf8]{inputenc}
\usepackage{graphicx}
\usepackage{booktabs}
\usepackage{algorithm, algpseudocode}
\usepackage[switch]{lineno}
\usepackage{amsmath,bbold,amsthm, amssymb, mathtools}
\usepackage{xcolor}
\usepackage{nicematrix}
\usepackage{mathtools}

\newcommand{\mnorm}[1]{{\left\vert\kern-0.25ex\left\vert\kern-0.25ex\left\vert #1 
    \right\vert\kern-0.25ex\right\vert\kern-0.25ex\right\vert}}

\newtheorem{definition}{Definition} 

\newtheorem{lemma}{Lemma}

\newtheorem{assumption}{Assumption}

\newtheorem{problem}{Problem}

\usepackage[caption=false]{subfig}
\usepackage[T1]{fontenc}
\usepackage{newtxtext}
\usepackage{newtxmath}

\usepackage{diagbox}
\usepackage{booktabs}
\usepackage[hidelinks]{hyperref}
\allowdisplaybreaks[2]

\usepackage{fancyhdr}
\fancypagestyle{firstpage}{%
  \fancyhf{}
  
  \fancyfoot[C]{\footnotesize This work has been submitted to the IEEE for possible publication. Copyright may be transferred without notice, after which this version may no longer be accessible.}
}

\begin{document}

\title{Allocating Corrective Control to Mitigate Multi-agent Safety Violations Under Private Preferences}

\author{
Johnathan Corbin\textsuperscript{1}, 
Sarah H.Q. Li\textsuperscript{1}, 
Jonathan Rogers\textsuperscript{1}%
\thanks{
\textsuperscript{1} All authors are with the Daniel Guggenheim School of Aerospace Engineering, Georgia Institute of Technology, Atlanta, GA 30332, USA. Emails: \texttt{jcorbin33@gatech.edu}, \texttt{sarahli@gatech.edu}, \texttt{jonathan.rogers@ae.gatech.edu}.
}%
}

\maketitle

\maketitle
\thispagestyle{firstpage} % Apply the footer to the first page only
\pagestyle{empty}         % Ensure subsequent pages stay clean

%Main body starts

\begin{abstract}
We propose a novel framework that computes the corrective control efforts to ensure joint safety in multi-agent dynamical systems. This framework efficiently distributes the required corrective effort without revealing  individual agents' private preferences. Our framework integrates high-order control barrier functions (HOCBFs), which enforce safety constraints with formal guarantees of safety for complex dynamical systems, with a privacy-preserving resource allocation mechanism based on the progressive second price (PSP) auction. When a joint safety constraint is violated, agents iteratively bid on new corrective efforts via `avoidance credits' rather than explicitly solving for feasible corrective efforts that remove the safety violation. The resulting correction, determined via a second price payment rule, coincides with the socially optimal safe distribution of corrective actions. Critically, the bidding process achieves this optimal allocation efficiently and without revealing private preferences of individual agents. We demonstrate this method through multi-robot hardware experiments on the Robotarium platform.
\end{abstract}
 
\section{Introduction}\label{sec:intro}
Ensuring the safe and reliable operation of autonomous multi-agent systems is a critical challenge in fields like transportation, logistics, and defense, especially when agents have distinct objectives. Inter-agent safety, which can range from phenomena like maintaining inter-agent separation distances \cite{wang_safety_2017} to ensuring communication quality \cite{de_carli_distributed_2024}, is crucial for operational success of multi-agent systems. A key problem is how to distribute the responsibility for safety maneuvers. While existing methods can guarantee safety, they often impose a disproportionate control effort on certain agents or neglect the individual preferences of agents. Over time, this imbalance can lead to system inefficiencies or strategic behavior from agents.

This work addresses this challenge by introducing a coordination framework that guarantees safety, achieves efficient control allocations, and preserves the privacy of agents' preferences. Our approach integrates high-order control barrier functions (HOCBFs) for provable safety guarantees \cite{borrmann_control_2015, wang_safety_2017, ames_control_2014, xiao_high-order_2022} with an auction mechanism for control allocation \cite{lazar2001design, ma_auction_2020}. Unlike methods that pre-allocate regions of the state space \cite{leet_combinatorial_2024}, our framework allocates the real-time corrective control effort required to satisfy the HOCBF, bridging the gap between HOCBFs and auctions. 

While HOCBFs produce safe corrective efforts with formal guarantees, they lack a formal mechanism for selecting feasible corrective efforts that align with private preferences. Conversely, auctions can guarantee that any corrective effort is optimally distributed among agents to align with their private preferences, but generally cannot guarantee that these allocated efforts ensure safety within a multi-agent dynamical system. Our method integrates these approaches to create a framework for the incentive-compatible and efficient allocation of corrective control effort in real-time multi-agent systems.

\textbf{Contributions.} We formulate the multi-agent HOCBF safety requirement as an allocable resource by quantifying the necessary corrective effort that must be shared. Second, we adapt the progressive second price (PSP) mechanism to this setting to allocate safety responsibility efficiently according to agents' private valuations. Finally, we integrate these components into a unified framework that ensures both provable safety and efficient control allocations. We demonstrate our algorithm  using hardware experiments of multi-robot interactions, performed using the Robotarium~\cite{wilson_robotarium_2020} and compare its performance in safety satisfaction, allocation fairness, and efficiency to conventional approaches.

This framework provides a novel approach to allocating control actions that jointly satisfy dynamical safety constraints,
with potential applications in cooperative robotics \cite{emam_robust_2019}, transportation , and airspace management \cite{xu_guaranteed_2025}. 

% \sarahl{This section can remove if tight on space} The remainder of this paper reviews related work (Sec. \ref{sec:related_res}), formulates the problem (Sec. \ref{sec:prelim}), presents our integrated framework (Sec. \ref{sec:theory}), details experimental results (Sec. \ref{sec:results}), and concludes (Sec. \ref{sec:conclusion}).
\section{Related Research}\label{sec:related_res}
Control barrier functions (CBFs) are a tool for guaranteeing the forward invariance of safe sets in dynamical systems, with applications spanning robotics, adaptive cruise control, and multi-agent collision avoidance \cite{ames_control_2019, ames_control_2014, wang_safety_2017}. In multi-agent settings, however, existing CBF formulations often assume equal effort sharing, predefine agent contributions \cite{borrmann_control_2015}, or learn them from data \cite{remy_learning_2025}, leaving open the problem of allocating safety responsibility based on agents' private preferences. For systems where the control input does not appear in the first derivative of the safety function (i.e., relative degree greater than one), HOCBFs have been introduced to extend these safety guarantees to a broader class of dynamics \cite{xiao_high-order_2022}.

Separately, auction theory is well-established for allocating resources among strategic agents with private valuations \cite{dias_market-based_2006, ma_auction_2020}. Preserving the privacy of these valuations is critical, as they can represent sensitive information related to an agent's mission priorities or competitive strategies \cite{ma_auction_2020}. These mechanisms have been applied to problems including path planning, bandwidth allocation, and airspace deconfliction \cite{calliess_lazy_2012, zou_resource_2018, leet_combinatorial_2024}. Within this domain, Vickrey-Clarke-Groves (VCG) mechanisms are frequently employed to guarantee incentive-compatible bidding strategies \cite{vickrey_counterspeculation_1961, clarke_multipart_1971, groves_incentives_1973}. To address auction allocations over repeated interactions, recent works have also investigated alternative frameworks such as karma systems \cite{elokda_self-contained_2024, berriaud_learning_2023}. A significant limitation, however, is that these approaches typically focus on discrete goods or task-level coordination, not continuous-time safety constraints.

\section{Preliminaries and Problem Formulation}
\label{sec:prelim}
This section establishes the mathematical foundation for our work. We first introduce the agent dynamics and the concept of private agent preferences. We then formally state the central problem of this paper before detailing the HOCBF framework used to enforce the safety constraints.

\paragraph{Notation}
We use $\mathbb{R}^n$ and $\mathbb{R}^n_+$ to denote the $n$-dimensional real and non-negative real vector spaces, respectively. A function $\alpha: \mathbb{R} \to \mathbb{R}$ is an extended class-$\mathcal{K}_\infty$ function if it is strictly increasing, $\alpha(0)=0$, and defined over all real numbers. We denote the Lie derivative of a differentiable function $h(\mathbf{x})$ along a vector field $f(\mathbf{x})$ as $L_f h(\mathbf{x}) \coloneqq \nabla_{\mathbf{x}} h(\mathbf{x})f(\mathbf{x})$. Higher-order Lie derivatives are defined recursively as $L_f^k h(\mathbf{x}) \coloneqq L_f(L_f^{k-1} h(\mathbf{x}))$. Similarly, for a matrix-valued function $G(\mathbf{x})$, we define $L_G h(\mathbf{x}) \coloneqq \nabla_{\mathbf{x}} h(\mathbf{x}) G(\mathbf{x})$ \cite{khalil_nonlinear_2002}.

\subsection{Agent Dynamics and Nominal Control}
We consider a multi-agent system of $N$ autonomous agents, indexed by $\mathcal{I} = \{1, \ldots, N\}$. The dynamics of each agent $i \in \mathcal{I}$ are modeled as a nonlinear, control-affine system:
\begin{equation}
    \dot{\mathbf{x}}_i = f_i(\mathbf{x}_i) + g_i(\mathbf{x}_i) \mathbf{u}_i,
    \label{eq:system}
\end{equation}
where $\mathbf{x}_i \in \mathbb{R}^{n_i}$ is the state, $\mathbf{u}_i \in \mathbb{R}^{m_i}$ is the control input, and the functions $f_i$ and $g_i$ are locally Lipschitz continuous. Let $\bar{n} = \sum_{i=1}^{N} n_i$ be the total dimension of the joint state space. The aggregate state is $\mathbf{x} = [\mathbf{x}_1^\top, \ldots, \mathbf{x}_N^\top]^\top \in \mathbb{R}^{\bar{n}}$, and the aggregate control is $\mathbf{u} = [\mathbf{u}_1^\top, \ldots, \mathbf{u}_N^\top]^\top$. The full system dynamics can be written compactly as:
\begin{equation}
    \dot{\mathbf{x}} = f(\mathbf{x}) + G(\mathbf{x})\mathbf{u},
    \label{eq:full_system}
\end{equation}
where $f(\mathbf{x})$ stacks the drift dynamics $f_i(\mathbf{x}_i)$ and $G(\mathbf{x})$ is a block-diagonal matrix of the control mappings $g_i(\mathbf{x}_i)$.

\begin{assumption}[Nominal Control]
Each agent $i \in \mathcal{I}$ has a \textbf{nominal control} $\hat{\mathbf{u}}_i \in \mathbb{R}^{m_i}$ that defines its preferred control to achieve its primary objective (e.g., trajectory tracking).
\end{assumption}

\subsection{Multi-agent Safety and Problem Formulation}
Safety for the multi-agent system is defined by a set of $M$ multi-agent safety functions, indexed by the set $\mathcal{K} = \{1, \ldots, M\}$. These functions can model a variety of safety metrics, such as collision avoidance \cite{palani_collision_2024} or maintaining communication connectivity \cite{de_carli_distributed_2024}. Each safety function is represented by a continuously differentiable function $h_k(\mathbf{x}, t)$ where the condition $h_k(\mathbf{x}, t) \ge 0$ implies safety. To ensure that these functions represent inter-agent safety, each $h_k$ must depend on the states of a subset of agents $\mathcal{I}_k \subseteq \mathcal{I}$, where $|\mathcal{I}_k| \ge 2$.

Global safety is achieved when all safety functions are simultaneously non-negative, i.e. $h_k(\mathbf{x}, t) \geq 0$, $\forall k \in \mathcal{K}$. We capture this with a global, time-varying safe set $\mathcal{C}(t)$:
\begin{equation}
    \mathcal{C}(t) = \{ \mathbf{x} \in \mathbb{R}^{\bar{n}} \!:\! H(\mathbf{x}, t) \geq 0 \}, \enspace H(\mathbf{x}, t) \coloneqq \min_{k \in \mathcal{K}} h_k(\mathbf{x}, t).
    \label{eq:safe_set}
\end{equation}
This function $H(\mathbf{x}, t)$ aggregates all individual safety functions into a single function, serving as a global safety indicator for the entire system. To ensure the system remains within this global safe set, we utilize the HOCBF framework described in Section \ref{sec:HOCBF}.

Maintaining the system within this safe set often requires agents to deviate from their nominal control $\hat{\mathbf{u}}_i$. We formally define this control deviation as:
$$ \boldsymbol{\delta}_i \coloneqq \mathbf{u}_i - \hat{\mathbf{u}}_i, $$
where $\mathbf{u}_i$ denotes the new control that agent $i$ must apply to maintain safety. Since agents are self-interested, they prefer to minimize this deviation. To manage the allocation of this burden, we introduce an abstract, divisible resource called \textit{avoidance credit}, denoted by $c_i \in [0,1]$. We define this as a finite shared resource, normalized such that the total credit distributed among the agents sums to 1 (i.e., $\sum_{i \in \mathcal{I}} c_i = 1$). Intuitively, $c_i$ represents the percentage of the total \textit{right of way} secured by agent $i$: a higher allocation grants the agent a greater privilege to maintain its nominal trajectory, thereby reducing its share of the necessary corrective effort. This credit represents a normalized relief from control deviation, establishing a clear inverse relationship: an agent receiving full credit ($c_i=1$) is permitted to apply zero corrective deviation ($\|\boldsymbol{\delta}_i\| \to 0$), while an agent with no credit ($c_i=0$) must bear a larger share of the required corrective action.

\begin{definition}[Private Valuation Function]
Each agent $i \in \mathcal{I}$ has a private valuation function $v_i: [0, 1] \to \mathbb{R}_+$ that quantifies its utility for receiving avoidance credit $c_i$.
\end{definition}

\begin{assumption}[Valuation Function Properties]
\label{as:valuation_func}
Each valuation function $v_i$ is assumed to be strictly increasing (i.e., $\frac{dv_i}{dc_i} > 0$) and strictly concave (i.e., $\frac{d^2v_i}{dc_i^2} < 0$).
\end{assumption}

Given these preferences, our core objective is to ensure safety while optimizing the distribution of the burden.

\begin{problem}[Optimal and Private Allocation]
\label{prob:main}
Find the socially optimal allocation of avoidance credit, $c^* = [c_1^*, \dots, c_N^*]^\top$, that solves the welfare maximization problem:
\begin{equation}
    \max_{c \in [0,1]^N} \; \sum_{i=1}^N v_i(c_i)
    \quad \text{s.t.} \quad \sum_{i=1}^N c_i = 1,
    \label{eq:social_opt_main}
\end{equation}
under the constraint that agents do not reveal their complete private valuation functions $(v_1, \dots, v_N)$.
\end{problem}

The solution to Problem~\ref{prob:main} yields an allocation that maximizes the collective utility of the system, respecting agents' subjective costs for deviating from their nominal behavior. The privacy constraint is critical for practical implementation, as it eliminates the need for agents to reveal sensitive internal cost information. We propose to solve this constrained optimization problem using an auction game, detailed in Section~\ref{sec:theory}.

\subsection{Safety Enforcement via HOCBFs}
\label{sec:HOCBF}
To strictly enforce the safety requirement in~\eqref{eq:safe_set} on the dynamics in~\eqref{eq:full_system}, we employ the HOCBF framework \cite{xiao_high-order_2022}. This is necessary when the control input $\mathbf{u}$ does not directly influence the first time-derivative of $H(\mathbf{x}, t)$, a situation formally captured by the concept of relative degree.

\begin{definition}[Relative Degree]
A function $H(\mathbf{x}, t)$ is said to have \textbf{relative degree} $m \in \mathbb{N}$ with respect to the system dynamics~\eqref{eq:full_system} in a domain $\mathcal{D} \subseteq \mathbb{R}^{\bar{n}}$ if for all $\mathbf{x} \in \mathcal{D}$ \cite{khalil_nonlinear_2002}:
\begin{align*}
    L_G L_f^k H(\mathbf{x}, t) &= 0 \quad \forall k < m-1, \\
    L_G L_f^{m-1} H(\mathbf{x}, t) &\neq 0.
\end{align*}
\end{definition}

In other words, the relative degree $m$ is the number of times $H(\mathbf{x}, t)$ must be differentiated with respect to time before the control input $\mathbf{u}$ explicitly appears. A standard CBF requires a relative degree of one. When $m > 1$, we can utilize the HOCBF framework which systematically differentiates the safety function until the $m$-th derivative, $H^{(m)}$, which is the first that can be directly shaped by the control input.

Because the $\min$ operator in $H(\mathbf{x}, t)$ is non-differentiable, we approximate it using the continuously differentiable log-sum-exp (LSE) function \cite{usevitch_adversarial_2023}:
\begin{equation}
    \tilde{H}(\mathbf{x}, t) \coloneqq -\frac{1}{\lambda} \log\!\left( \sum_{k \in \mathcal{K}} \exp\!\big(-\lambda\, h_k(\mathbf{x},t)\big) \right),
\label{eq:lse_approx}
\end{equation}
where $\lambda > 0$ is a sharpness parameter. The LSE function is a smooth under-approximation, i.e., $\tilde{H}(\mathbf{x}, t) \le H(\mathbf{x}, t)$, ensuring a conservative safety margin \cite{lindemann_control_2019}.

To account for the higher relative degree, we define a sequence of auxiliary functions $\psi_i$ recursively as $\psi_0 = \tilde{H}$ and $\psi_i = \dot{\psi}_{i-1} + \alpha_i(\psi_{i-1})$ for $i \in \{1, \dots, m\}$, where each $\alpha_i$ is an extended class-$\mathcal{K}_\infty$ function. We also define the term $O(\tilde{H})$ to denote the Lie derivatives along $f$ and partial derivatives with respect to $t$ that arise from the chain rule with a degree less than or equal to $m-1$:
\begin{equation}
    O(\tilde{H}) = \sum_{j=1}^{m-1} \left[ \begin{multlined}
        L_f^j (\alpha_{m-j}(\psi_{m-j-1})) \\
        + \frac{\partial^j (\alpha_{m-j}(\psi_{m-j-1}) )}{\partial t^j} \end{multlined} \right].
\end{equation}

\begin{lemma}[Sufficient Condition for Forward Invariance]
\label{lem:sufficient_condition}
Let $\tilde{H}(\mathbf{x}, t)$ be the LSE approximation \eqref{eq:lse_approx} of $H(\mathbf{x}, t)$ \eqref{eq:safe_set} with relative degree $m$. The safe set $\mathcal{C}(t)$ \eqref{eq:safe_set} is rendered forward invariant $\forall t \geq 0$ if the joint control $\mathbf{u}(t)$ satisfies the affine constraint
\begin{equation}
    A(\mathbf{x},t)\mathbf{u}(t) \geq b(\mathbf{x},t), \quad \forall t \geq 0,
\label{eq:global_safety}
\end{equation}
where $A \coloneqq L_G L_f^{m-1}\tilde{H}$ and $b \coloneqq -L_f^m\tilde{H}-\frac{\partial^m \tilde{H}}{\partial t^m}-O(\tilde{H})-\alpha_m (\psi_{m-1})$.
\end{lemma}

\begin{proof}
Recall that the LSE function is a smooth under-approximation, satisfying $\tilde{H}(\mathbf{x}, t) \le H(\mathbf{x}, t)$ for all $(\mathbf{x}, t) \in \mathbb{R}^{\bar{n}} \times \mathbb{R}_+$. This implies that the superlevel set of $\tilde{H}$ is contained within the safe set $\mathcal{C}(t)$. Therefore, according to the HOCBF framework \cite{xiao_high-order_2022}, rendering the superlevel set of $\tilde{H}$ forward invariant is sufficient for rendering $\mathcal{C}(t)$ forward invariant. The condition for this is:
$$ L_f^m \tilde{H} + L_G L_f^{m-1} \tilde{H} \mathbf{u} + \frac{\partial^m \tilde{H}}{\partial t^m} + O(\tilde{H}) + \alpha_m (\psi_{m-1}) \geq 0. $$
Rearranging this inequality yields the affine constraint $A\mathbf{u} \ge b$ as defined in the lemma statement.
\end{proof}
This linear constraint in~\eqref{eq:global_safety} is the core safety requirement that any valid control law $\mathbf{u}$ must satisfy in order to maintain forward invariance of the safe set $\mathcal{C}(t)$.
\section{Safety Allocation Mechanism}
\label{sec:theory}

To solve Problem~\ref{prob:main}, we first translate the global HOCBF safety constraint into a quantifiable resource representing each agent's share of the safety burden. We then introduce an auction game that allocates this responsibility in a way that is socially optimal, respects agent privacy, and guarantees safety.

\subsection{Safety Deficit and Responsibility Allocation}
When the multi-agent system is operating safely under nominal control, the condition $A\hat{\mathbf{u}} \geq b$ holds, and no intervention is required. A safety-critical event occurs when this condition is violated. We quantify this violation using the \textit{safety deficit}:
\begin{equation}
    S(\mathbf{x}, t) \coloneqq b - A\hat{\mathbf{u}}.
    \label{eq:safety_deficit}
\end{equation}
When $S > 0$, the nominal controls $\hat{\mathbf{u}}$ are unsafe, and a total corrective effort of at least $S$ must be applied by the system. This requirement can be expressed in terms of the total safety correction, $\Delta \coloneqq A(\mathbf{u}-\hat{\mathbf{u}})$, leading to the compact safety condition $\Delta \geq S$.

Since the system is composed of multiple agents, this total safety correction $\Delta$ is the sum of individual contributions:
\begin{equation}
    \Delta = \sum_{i \in \mathcal{I}} \Delta_i.
\end{equation}
Each scalar term $\Delta_i \coloneqq A_i (\mathbf{u}_i - \hat{\mathbf{u}}_i) = A_i \boldsymbol{\delta}_i$ represents the individual contribution of agent $i$ to the total safety correction. Physically, $\Delta_i$ corresponds to the change in the high-order time derivative of the safety function $\tilde{H}(\mathbf{x}, t)$ induced by agent \textit{i}'s deviation from its nominal control. The row vector $A_i$ consists of the columns of the matrix $A$ corresponding to agent $i$'s control inputs, $\mathbf{u}_i$.
The global safety constraint can then be written as a shared responsibility among all agents:
\begin{equation}
    \sum_{i \in \mathcal{I}} \Delta_i \geq S.
\label{eq:new_safety}
\end{equation}
This formulation frames safety as an allocable burden. In practice, only a subset of agents, $\mathcal{I}_{\text{active}} \subseteq \mathcal{I}$, may be directly responsible for a potential safety violation. To identify these agents, we evaluate the HOCBF safety condition individually for each safety function $h_k(\mathbf{x}, t)$ using the nominal controls $\hat{\mathbf{u}}$. If the condition
$$ L_f^m h_k + L_G L_f^{m-1} h_k \hat{\mathbf{u}} + \frac{\partial^m h_k}{\partial t^m} + O(h_k) + \alpha_m (\psi_{m-1}) < 0 $$
is met for a given $h_k$, then that safety function is considered active. All agents whose states are arguments of any active safety function are added to the set $\mathcal{I}_{\text{active}}$. We note that this approach may group agents from independent conflict scenarios (i.e., violations of safety functions involving disjoint sets of agents) into a single auction. For the scope of this work, we treat this as a single resolution event, while partitioning $\mathcal{I}_{\text{active}}$ into decoupled subgroups is a topic for future work.

To ensure computational efficiency, the auction mechanism is not invoked at every time step. An auction is triggered only when a significant change in the safety landscape occurs. Specifically, the conditions for triggering a new auction are: 1) the system would transition from safe to unsafe state under nominal control (i.e., the safety deficit $S$ becomes positive), or 2) the composition of the active set $\mathcal{I}_{\text{active}}$ changes from the previous time step. This event-triggered approach avoids redundant computation while the set of interacting agents remains constant.

To connect the abstract avoidance credit $c_i$ from Section~\ref{sec:prelim} to the required safety correction $\Delta_i$, we define the following conversion. Because avoidance credit is a unit-valued resource distributed among the active agents, any valid allocation must satisfy $\sum_{i \in \mathcal{I}_{\text{active}}} c_i = 1$. The safety correction $\Delta_i$ for an agent receiving credit $c_i$ is then given by:
\begin{equation}
    \Delta_i(c_i) = \frac{(1 - c_i) S}{\sum_{j \in \mathcal{I}_{\text{active}}} (1 - c_j)}.
    \label{eq:credit_mapping}
\end{equation}
This mapping ensures that the total required safety correction is always met, i.e., $\sum_{i \in \mathcal{I}_{\text{active}}} \Delta_i(c_i) = S$. It provides an intuitive link between credit and responsibility: an agent with full credit ($c_i=1$) has zero corrective responsibility ($\Delta_i=0$), while agents with less credit proportionally share the burden of satisfying the safety deficit.

An important consequence of this formulation is that it greatly reduces the required frequency of auctions. Between auction events, each agent retains its last assigned avoidance credit share, $c_i^*$. However, because the total safety deficit $S(\mathbf{x}, t)$ evolves continuously with the system state, the actual safety correction required from each agent, $\Delta_i^*$, also adapts dynamically according to Eq.~\eqref{eq:credit_mapping}. This design forces agents to continuously adjust their control effort to meet the current safety demand without needing to run a new auction at every time step. This computational efficiency comes at the cost of temporal uncertainty: agents bid at the start of a safety violation without knowing the full cumulative effort required to resolve the violation. Because the framework does not currently predict the future evolution of $S(\mathbf{x}, t)$, agents effectively commit to a share of the safety burden without a guarantee of the total cost, a limitation we aim to address in future work by incorporating prediction horizons into the bidding process.

\subsection{An Auction Game for Optimal Allocation}
To solve Problem~\ref{prob:main}, we implement an auction game based on the PSP auction mechanism. The PSP auction is a method for resource allocation that uses a Vickrey-Clarke-Groves (VCG) payment rule to produce socially optimal resource allocations without revealing agents' private incentives~\cite{lazar2001design,ma_auction_2020}.

Before defining the game, we first introduce the PSP allocation and pricing mechanism. Given a bid profile $\mathbf{b} = (\mathbf{b}_i)_{i \in \mathcal{I}_{\text{active}}}$, where each bid $\mathbf{b}_i = (\beta_i, d_i)$ consists of a unit price $\beta_i \in \mathbb{R}_+$ and demand $d_i \in \mathbb{R}_+$ for the avoidance credit, the PSP mechanism computes the allocation $c_i \in \mathbb{R}_+$ and price $\pi_i \in \mathbb{R}_+$ for each agent as follows:

\begin{itemize}
    \item \textbf{Allocation Rule:} Given a profile of bids $\mathbf{b}$, the auctioneer determines the allocation $c(\mathbf{b})$ by solving the following optimization problem:
    \begin{equation}
        \begin{aligned}
            & \underset{c \in \mathbb{R}^{|\mathcal{I}_{\text{active}}|} }{\text{maximize}}
            & & \sum_{i \in \mathcal{I}_{\text{active}}} \beta_i c_i \\
            & \text{subject to}
            & & \sum_{i \in \mathcal{I}_{\text{active}}} c_i = 1, \\
            & & & c_i \le d_i, \quad \forall i \in \mathcal{I}_{\text{active}}.
        \end{aligned}
    \end{equation}
    
    \item \textbf{Payment Rule:} Each agent $i$ pays an amount $\pi_i(\mathbf{b})$ determined by the VCG rule~\cite{vickrey_counterspeculation_1961, clarke_multipart_1971, groves_incentives_1973}, which equals the declared welfare loss its participation causes to others (i.e. the externality):
    \begin{equation}
        \pi_i(\mathbf{b}) = \sum_{j \neq i} \beta_j (c_j(\mathbf{b}_{(i)}) - c_j(\mathbf{b})),
        \label{eq:vcg_payment}
    \end{equation}
    where $\mathbf{b}_{(i)} = ((\beta_i, 0), \mathbf{b}_{-i})$ is the bid profile when agent $i$ is considered absent (i.e., its demand $d_i$ is set to zero), and $\mathbf{b}_{-i}$ denotes the bids of all other agents~\cite{ma_auction_2020}.
\end{itemize}

We now formally define the game played by the agents under this mechanism.

\begin{definition}[Avoidance Credit Auction Game]
\label{def:auction_game}
An ACAG, $\mathcal{G} = (\mathcal{I}_{\text{active}}, \{\mathcal{B}_i\}_{i \in \mathcal{I}_{\text{active}}}, \{U_i\}_{i \in \mathcal{I}_{\text{active}}})$, is played by $\mathcal{I}_{\text{active}}$ players where each player $i$ has action (bid) set $ \mathcal{B}_i = [0, \infty)\times [0,1]$. The joint action space is denoted by $\boldsymbol{\mathcal{B}} \coloneqq \prod_{j \in \mathcal{I}_{\text{active}}} \mathcal{B}_j$, and each player has a utility function $U_i: \boldsymbol{\mathcal{B}} \to \mathbb{R}$.

Each player $i$'s action $\mathbf{b}_i = (\beta_i, d_i)$ consists of $\beta_i \in [0, \infty)$, player $i$'s price per unit of avoidance credit, and $d_i \in [0,1]$, its maximum credit quantity requested. Each player's utility $U_i$ is defined as:
\begin{equation}
    U_i(\mathbf{b}) = v_i(c_i(\mathbf{b})) - \pi_i(\mathbf{b}), \quad \forall \mathbf{b} \in \boldsymbol{\mathcal{B}},
\end{equation}
where $c_i(\mathbf{b})$ and $\pi_i(\mathbf{b})$ are the allocation and payment determined by the PSP mechanism for the joint bid profile $\mathbf{b}$.
\end{definition}

In this game, agents are effectively competing for the \textit{right of way}. A bid represents an agent's willingness to pay to maintain its nominal course. An agent that wins the auction (secures a high $c_i$) is granted a larger share of avoidance credit, which translates via \eqref{eq:credit_mapping} into a smaller required control deviation. Conversely, agents who lose the auction (receive low $c_i$) take on more of the burden of the evasive maneuver. Thus, the auction acts as a market mechanism where agents with large private valuations can pay to secure priority, effectively shifting the maneuvering burden to agents with lower private valuations.

The goal for each self-interested agent is to choose a bid $\mathbf{b}_i$ that maximizes its own utility $U_i$. The stable outcome of this game is described by the Nash equilibrium.

\begin{definition}[Nash Equilibrium]
A bid profile $\mathbf{b}^* = (\mathbf{b}_i^*, \mathbf{b}_{-i}^*)$ is a Nash equilibrium if no agent $i$ can improve its utility by unilaterally changing its bid, i.e., for every agent $i \in \mathcal{I}_{\text{active}}$:
$$ U_i(\mathbf{b}_i^*, \mathbf{b}_{-i}^*) \geq U_i(\mathbf{b}_i, \mathbf{b}_{-i}^*) \quad \forall \mathbf{b}_i \in \mathcal{B}_i. $$
\end{definition}

To find the Nash equilibrium, we adopt the iterative best-response dynamic from the PSP auction literature~\cite{ma_auction_2020, lazar2001design}. We emphasize that this is an established solution method, not a novel contribution. In this process, agents asynchronously update their bids; at each iteration, a single player computes and submits a new bid that maximizes its personal utility $U_i$ given the current bids of the other players. Under the properties in Assumption~\ref{as:valuation_func}, this dynamic process is guaranteed to converge to a unique Nash equilibrium~\cite{ma_auction_2020}.

Furthermore, a key property of this VCG-based mechanism is that it is incentive-compatible. This means it is a weakly dominant strategy for each agent to engage in truthful bidding, where it reports its true preferences~\cite{ma_auction_2020, lazar2001design}. In the context of this game, truthful bidding means an agent sets the price component of its bid, $\beta_i$, equal to its true marginal valuation at its demanded quantity, $d_i$, i.e., $\beta_i = \frac{dv_i}{dc_i}(d_i)$. When agents employ this optimal strategy, the Nash equilibrium reached by the iterative process coincides with the socially optimal allocation $c^*$ that solves Problem~\ref{prob:main}. Thus, by adopting the PSP auction, we achieve a private and socially optimal solution to our safety allocation problem.

\subsection{Control Synthesis from Optimal Allocation}
Once the auction game converges to the optimal allocation $c^*$, the result is mapped back to a concrete control action for each agent. First, the individual safety corrections $\Delta_i^*$ are computed using~\eqref{eq:credit_mapping}.

Each agent $i \in \mathcal{I}_{\text{active}}$ must then compute a new control input $\mathbf{u}_i$ that satisfies its assigned share of the safety burden:
$$ A_i (\mathbf{u}_i - \hat{\mathbf{u}}_i) \geq \Delta_i^*. $$
While multiple control inputs can satisfy this inequality, a simple approach is to find the one that does so with minimum effort by computing the minimum-norm solution using the Moore-Penrose pseudo-inverse:
\begin{equation}
    \mathbf{u}_i = \hat{\mathbf{u}}_i + A_i^\dag \Delta_i^*,
    \label{eq:control_update}
\end{equation}
where $A_i^\dag$ is the pseudo-inverse of $A_i$. This provides a unique and consistent method for agents to compute their safety-compliant control actions, though they retain the flexibility to choose any control that meets the constraint $A_i (\mathbf{u}_i - \hat{\mathbf{u}}_i) \geq \Delta_i^*$ in any way they see fit.
\section{Experimental Demonstration} \label{sec:results}

We demonstrate our approach using the Robotarium multi-robot testbed, which provides a platform with differential-drive robots that closely approximate unicycle dynamics \cite{wilson_robotarium_2020}. To specifically test the HOCBF aspect, we constrain the robots to have a constant forward velocity, allowing control only through the angular velocity. This deliberate design choice ensures that pairwise distance safety functions have a relative degree of two, highlighting the practical feasibility of using HOCBFs in our method.

\subsection{Experimental Model}
The experiments involve four agents with unicycle dynamics. Each agent $i$ has state $\mathbf{x}_i = [x_i \; y_i \; \theta_i]^\top$, controlled via its angular velocity $\omega_i$, with a fixed linear speed $v$. The kinematic model is given by:
\begin{equation}
    \dot{x}_i = v \cos\theta_i, \quad
    \dot{y}_i = v \sin\theta_i, \quad
    \dot{\theta}_i = \omega_i.
\end{equation}
This can be written in the standard control-affine form $\dot{\mathbf{x}}_i = f_i(\mathbf{x}_i) + g_i(\mathbf{x}_i)\omega_i$, where the drift dynamics $f_i(\mathbf{x}_i)$ and control vector field $g_i(\mathbf{x}_i)$ are:
\begin{equation}
    f_i(\mathbf{x}_i) =
    \begin{bmatrix}
        v\cos\theta_i \\ v\sin\theta_i \\ 0
    \end{bmatrix},
    \quad
    g_i(\mathbf{x}_i) =
    \begin{bmatrix}
        0 \\ 0 \\ 1
    \end{bmatrix}.
\end{equation}
Each agent follows a nominal control policy that combines a goal-seeking behavior in the $x$ direction with a line-keeping term in $y$. Specifically, each agent is assigned a reference line $y_{\mathrm{ref},i}$ (corresponding to its initial $y$ position), and the nominal heading is determined by a vector field that pulls the agent toward $(x_i^{\text{goal}}, y_{\mathrm{ref},i})$:
\begin{equation}
    \theta_i^{\mathrm{des}} = \mathrm{atan2}\!\left(-k_y (y_i - y_{\mathrm{ref},i}),\; k_x (x_i^{\text{goal}} - x_i)\right),
\end{equation}
where $k_x, k_y > 0$ are proportional gains. The nominal angular velocity is then a proportional controller that drives the agent's heading to this desired value:
\begin{equation}
    \hat{\omega}_i = k_\theta \big(\theta_i^{\mathrm{des}} - \theta_i\big),
\end{equation}
with $k_\theta > 0$ as a heading gain. This control law is used to compute the nominal input $\hat{\mathbf{u}}$ for all agents. The specific parameters used for the experiments are summarized in Table~\ref{tab:params}.

\begin{table}[!htbp]
\caption{Experimental Parameters}
\label{tab:params}
\centering
% Agent-Specific Parameters (Using your updated values)
\begin{tabular}{c c c c}
\hline
\textbf{Agent} & \textbf{Initial Pos.} ($x,y$) & \textbf{Goal Pos.} ($x,y$) & \textbf{$\alpha_i$} (Base Value) \\
\hline
1 & $(-1.5, 0.0)$ & $(0.5, 0.0)$ & 1.0 \\
2 & $(-0.5, -0.1)$ & $(-1.5, -0.1)$ & 1.0 \\
3 & $(0.5, -0.1)$ & $(-1.0, -0.1)$ & 1.0 \\
4 & $(1.5, -0.1)$ & $(-0.5, -0.1)$ & 1.0 \\
\hline
\end{tabular}

\vspace{4pt} % Adds a little space between tables

% Shared Parameters (Using your updated values)
\begin{tabular}{l r}
\hline
\multicolumn{2}{c}{\textbf{Shared System \& Controller Parameters}} \\
\hline
Linear Speed ($v$) & 0.1 m/s \\
Safety Distance ($d$) & 0.12 m \\
Controller Gains ($k_x, k_y, k_\theta$) & \{0.5, 2.5, 2.0\} \\
HOCBF Gains ($\kappa_1, \kappa_2$) & \{1.2, 1.2\} \\
Valuation Shape ($k$) & 5.0 \\
Valuation Scaling ($\gamma$) & 8.0 \\
\hline
\end{tabular}
\end{table}

\subsection{HOCBF Constraint Derivation}
\label{sec:hocbf_derivation}
Safety is defined by a set of pairwise distance safety functions. For each unique pair of agents $(i, j)$, we define a function, indexed by $k$, via $h_k(\mathbf{x}) = \|\mathbf{p}_i - \mathbf{p}_j\|^2 - d^2 \ge 0$, where $\mathbf{p}_i = [x_i, y_i]^\top$. To apply the HOCBF framework, we compute the Lie derivatives of $h_k$ with respect to the aggregate system dynamics. The first Lie derivative with respect to the control matrix $G(\mathbf{x})$ is zero, $L_G h_k(\mathbf{x}) = \mathbf{0}$, confirming a relative degree $m > 1$. The first derivative with respect to the drift dynamics $f(\mathbf{x})$ is:
\begin{equation}
\begin{aligned}
    L_f h_k(\mathbf{x}) = 2\big[ &(x_i - x_j)(v\cos\theta_i - v\cos\theta_j) \\
    &+ (y_i - y_j)(v\sin\theta_i - v\sin\theta_j) \big].
\end{aligned}
\end{equation}
The second derivative with respect to the drift is:
\begin{equation}
\begin{aligned}
    L_f^2 h_k(\mathbf{x}) = 2\big[ &(v\cos\theta_i - v\cos\theta_j)^2 \\
    &+ (v\sin\theta_i - v\sin\theta_j)^2 \big],
\end{aligned}
\end{equation}
and the mixed Lie derivative $L_G L_f h_k(\mathbf{x})$ is a row vector that is zero everywhere except for the entries corresponding to the controls of agents $i$ and $j$. These non-zero entries are:
\begin{align*}
    [L_G L_f h_k(\mathbf{x})]_i &= 2v \big[ -(x_i - x_j)\sin\theta_i + (y_i - y_j)\cos\theta_i \big] \\
    [L_G L_f h_k(\mathbf{x})]_j &= 2v \big[ (x_i - x_j)\sin\theta_j - (y_i - y_j)\cos\theta_j \big].
\end{align*}
This confirms that the relative degree is $m=2$. We use linear class-$\mathcal{K}_\infty$ functions, $\alpha_1(s) = \kappa_1 s$ and $\alpha_2(s) = \kappa_2 s$, where $\kappa_1, \kappa_2 > 0$ are positive constants. Substituting these into the general HOCBF condition yields the final form for each active pair:
\begin{equation}
\label{eq:hocbf_individual}
\begin{aligned}
    L_f^2 h_k + L_G L_f h_k \mathbf{u} + (\kappa_1+\kappa_2) L_f h_k + \kappa_1 \kappa_2 h_k \ge 0.
\end{aligned}
\end{equation}
This inequality is evaluated for each pair under the nominal control policy $\hat{\mathbf{u}}$ to determine which agents should be added to the active set, $\mathcal{I}_{\text{active}}$. To derive the single safety constraint used by the auction, the individual barrier functions $h_k$ are first aggregated into a global barrier function $\tilde{H}$ using the log-sum-exp approximation. The HOCBF framework is then applied to $\tilde{H}$ to derive the global affine safety constraint $A\mathbf{u} \ge b$.

\subsection{Valuation Function Design}
Each agent $i$ uses a private valuation function $v_i(c_i)$ for avoidance credit $c_i$, satisfying Assumption~\ref{as:valuation_func}. To explore how preferences may change dynamically, we model valuation based on past interactions:
\begin{equation}
    v_i(c_i) = \gamma^{n_i}\,\alpha_i \left(1 - e^{-k c_i}\right), \label{eq:exp_valuation}
\end{equation}
where $n_i$ counts the number of distinct auction events (i.e., collision encounters) agent $i$ has been involved in, and the parameters $\gamma, \alpha_i, k$ are defined in Table~\ref{tab:params}. This design is motivated by the idea that agents may become more willing to bid higher as they are repeatedly required to deviate from their nominal paths, thus reducing their required deviation in later collisions.

\subsection{Controllers for Comparison}
To isolate and evaluate the contribution of our auction mechanism, we compare two safety filter implementations. Both controllers are implemented using the same HOCBF formulation derived in Section~\ref{sec:hocbf_derivation}, utilizing the same dynamics models and identical $\kappa$ gains (as defined in Table~\ref{tab:params}). As such, both controllers generate the same global safety constraint $A\mathbf{u} \ge b$. The only difference between them is the strategy used to solve for the final control $\mathbf{u}$.

\paragraph{Baseline: Standard HOCBF Controller}
This controller represents the standard, non-preference-aware method for resolving HOCBF constraints. It solves a single quadratic program (QP) that finds a control $\mathbf{u}$ that minimally alters (in an $L_2$-norm sense) the nominal controls $\hat{\mathbf{u}}$ while satisfying the global safety constraint. This approach is practically identical to the Robotarium's default supervisor \cite{emam_robust_2019}, albeit with a different formulation for the CBF constraint:
\begin{equation}
    \mathbf{u}^* = \arg\min_{\mathbf{u}} \sum_{i \in \mathcal{I}} \|\mathbf{u}_i - \hat{\mathbf{u}}_i\|^2 \quad \text{s.t.} \quad A\mathbf{u} \ge b.
\end{equation}
This QP-based approach is impartial as it does not consider agent-specific priorities. The minimization of the sum of squared errors implicitly distributes the corrective burden based on the agents' relative control authority over the constraint, not their preferences.

\paragraph{Proposed: Auction-Based HOCBF Controller}
Our proposed controller does not solve the above QP. Instead, it computes the safety deficit $S(\mathbf{x}, t) = b(\mathbf{x}, t) - A(\mathbf{x}, t)\hat{\mathbf{u}}(\mathbf{x})$ and uses the PSP auction to allocate this burden. Agents in $\mathcal{I}_{\text{active}}$ bid for avoidance credit $c_i$ based on their private, dynamic valuation functions~\eqref{eq:exp_valuation}. The resulting allocation $c^*$ is translated into individual safety corrections $\Delta_i^*$~\eqref{eq:credit_mapping}, which are then implemented using the pseudo-inverse~\eqref{eq:control_update}.

\subsection{Results and Analysis}
We implement the 4-agent collision scenario defined by the parameters in Table~\ref{tab:params}. In this setup, all agents are assigned an identical base valuation ($\alpha_i = 1.0$). This ensures that agents have symmetric preferences at the start of the experiment, allowing us to isolate the effect of the dynamic valuation term $\gamma^{n_i}$ as collision histories differ between agents.

\begin{figure}[!htbp]
    \centering
    \includegraphics[width=\columnwidth]{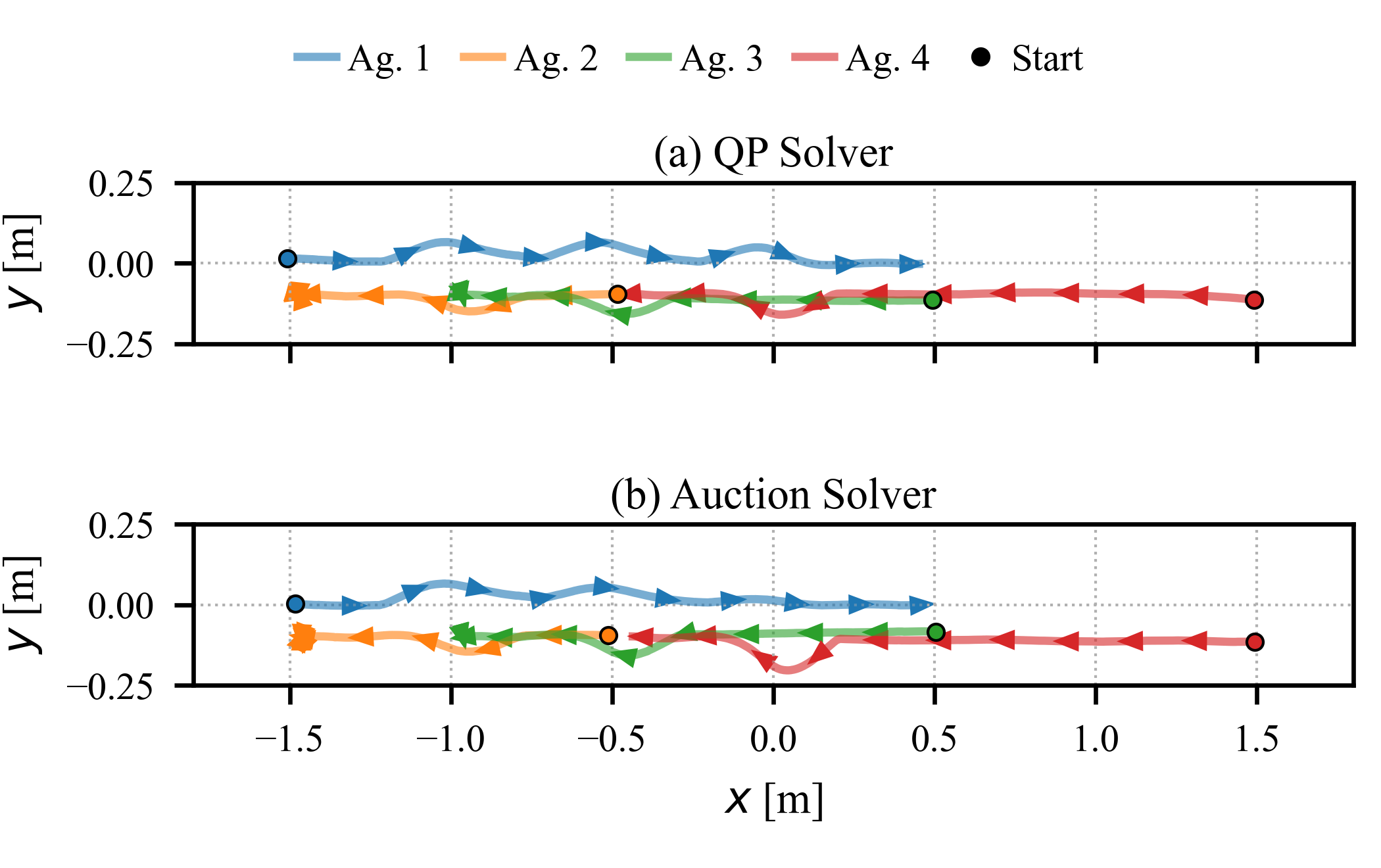}
    \caption{Agent trajectories for the 4-agent collision scenario with arrows indicating direction of travel. (a) With the baseline QP-HOCBF controller, Agent 1 (blue) is forced to deviate significantly in all three collision events. (b) With our proposed Auction-HOCBF controller, Agent 1 deviates significantly in the first collision, but its dynamically increasing valuation function causes other agents to bear more of the burden in subsequent collisions.}
    \label{fig:trajectories}
\end{figure}

\begin{figure}[!hbp]
    \centering
    \includegraphics[width=\columnwidth]{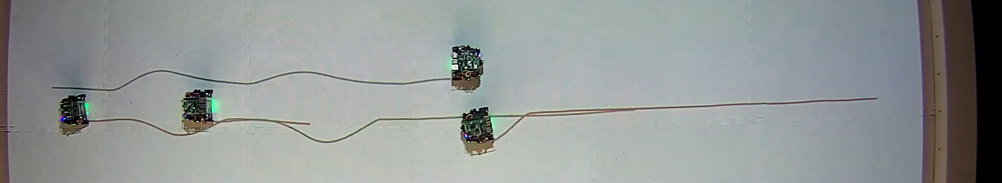}
    \caption{Snapshot of the hardware experiment on the Robotarium platform for the system utilizing the Auction-HOCBF controller during the third collision event.}
    \label{fig:experiment}
\end{figure}

Both controllers successfully guaranteed system safety throughout the experiment. The global safety function remained non-negative ($\tilde{H}(\mathbf{x},t) \ge 0$) at all times for both cases, confirming that the minimum separation distance $d$ was never violated. Importantly, because both controllers utilize the exact same HOCBF constraint formulation (derived in Section~\ref{sec:hocbf_derivation}), any differences in agent trajectories and control effort can be attributed directly to the allocation mechanism (QP vs. Auction) rather than to discrepancies in the safety conditions.

The trajectories of the agents under both controllers is shown in Fig.~\ref{fig:trajectories}. Additionally, Fig.~\ref{fig:experiment} shows the hardware experiment being performed on the Robotarium. With the baseline QP controller (Fig.~\ref{fig:trajectories}a), Agent 1 (blue) is forced to deviate significantly in all three distinct collision events it encounters. This occurs because the QP is impartial and solves the same myopic $L_2$-norm minimization problem for each collision, treating every event in isolation without regard for an agent's cumulative effort or prior burden.

% \begin{figure*}[t]
%     \centering
%     % --- Top Row: Start, 1st Collision, 2nd Collision ---
%     \subfloat[Initial Positions ($t=0$s)]{
%         % trim={<left> <lower> <right> <upper>} - Adjust if needed to crop black borders
%         \includegraphics[trim={0 0 0 0}, clip, width=0.3\textwidth]{figures/fig_start.png}
%         \label{fig:exp_start}
%     }
%     \hfil
%     \subfloat[Collision 2: Ag 1 vs 3]{
%         \includegraphics[trim={0 0 0 0}, clip, width=0.3\textwidth]{figures/fig_2nd.png}
%         \label{fig:exp_col2}
%     }
%     \hfil
%     \subfloat[Final Positions]{
%         \includegraphics[trim={0 0 0 0}, clip, width=0.3\textwidth]{figures/fig_final.png}
%         \label{fig:exp_end}
%     }
    
%     \caption{Hardware experiment snapshots utilizing the Auction-HOCBF controller. (a) Robots initialize at opposing ends. (b)Agent 1 and Agent 3 during the second conflict. (e) All agents reach their goals safely.}
%     \label{fig:hardware_snapshots}
% \end{figure*}

In contrast, our proposed auction-based controller (Fig.~\ref{fig:trajectories}b) exhibits a distinct shift in burden driven by interaction history. Agent 1 is involved in the first collision (with Agent 2), which increments its collision counter $n_1$. In the subsequent collisions, its valuation for avoidance credit is now higher than its peers. Agent 1 therefore wins more avoidance credit in future auctions and forces the other agents to perform a greater share of the avoidance maneuver.

\begin{figure}[!htbp]
    \centering
    \includegraphics[width=\columnwidth]{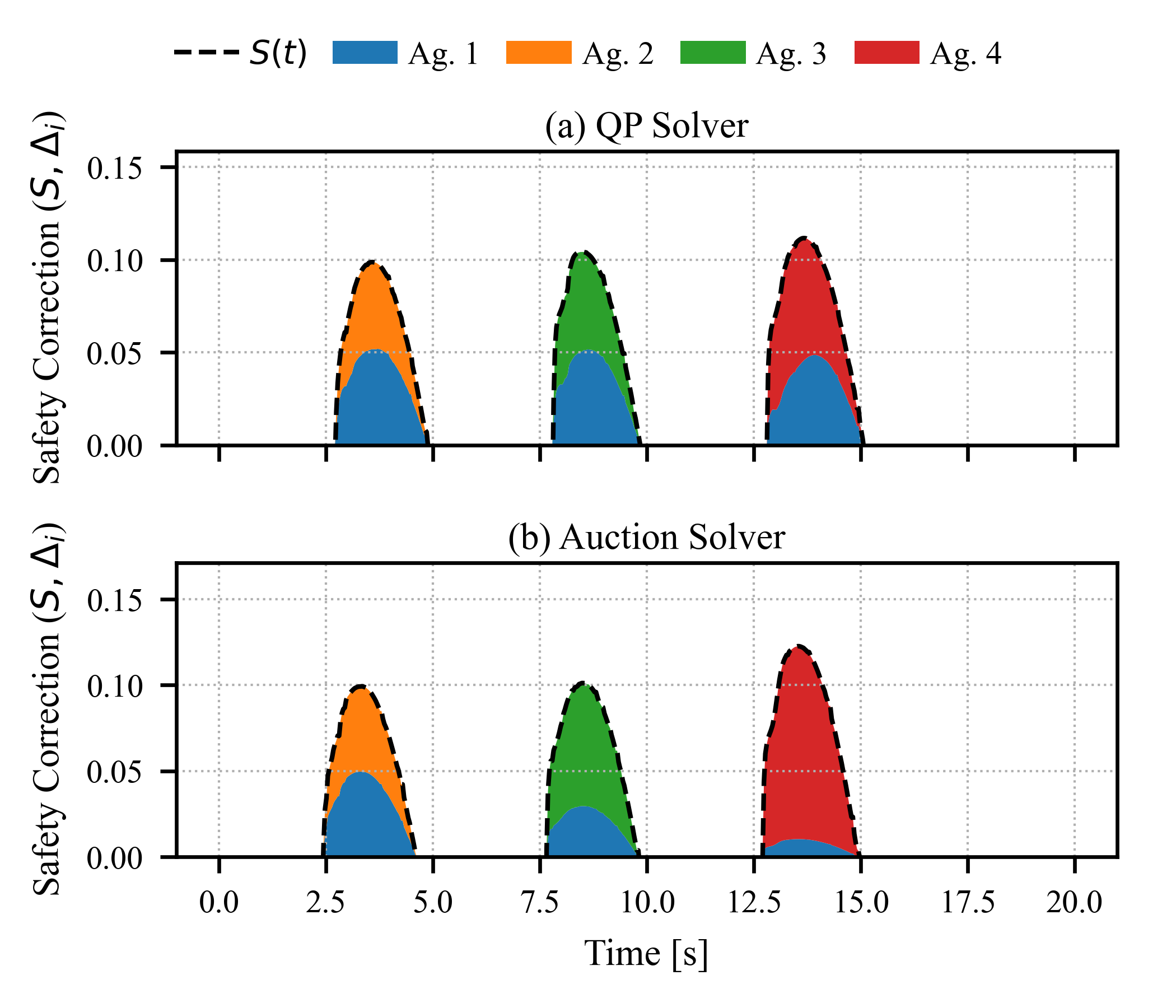}
    \caption{Stacked area plot of safety correction ($\Delta_i$) over time. The combined height of the stacks corresponds to the instantaneous safety deficit $S(t)$ (black dashed line). (a) The baseline QP-HOCBF divides the effort in a similar way for each collision. (b) Our Auction-CBF dynamically shifts the effort allocation in later collisions as Agent 1's valuation increases.}
    \label{fig:effort_plot}
\end{figure}

This difference in allocation is quantified in Fig.~\ref{fig:effort_plot}. The stacked area plot shows the distribution of the safety correction $\Delta_i$ over time, where the total height at any instant sums to the safety deficit $S(t)$ at that time. The resulting auction allocations of avoidance credit for each of the collision events highlights that Agent 1 starts to win more avoidance credit in later auctions as its valuation increases:
\begin{itemize}
    \item \textbf{Event 1 (Agent 1 vs. 2):} With $n_1=0, n_2=0$, valuations are equal. The auction allocates credit $c_1=0.50$ and $c_2=0.50$, resulting in a shared burden.
    \item \textbf{Event 2 (Agent 1 vs. 3):} Agent 1 now has $n_1=1$. Its higher valuation wins it $c_1=0.7079$ of the credit, forcing Agent 3 to take $c_3=0.2921$.
    \item \textbf{Event 3 (Agent 1 vs. 4):} With $n_1=2$, Agent 1's valuation is even higher, and it wins $c_1=0.9159$ of the credit, leaving only $c_4=0.0841$ for Agent 4.
\end{itemize}

The long-term impact of this mechanism is quantified in Table~\ref{tab:cumulative_effort}, which lists the final cumulative control effort for each agent. This quantity is calculated as the time integral of the absolute deviation between the applied control input and the nominal control input, $\int \|\mathbf{u}_i(t) - \hat{\mathbf{u}}_i(t)\| dt$. Physically, this represents the total amount of corrective turning (in radians) each agent had to perform to maintain safety.

\begin{table}[!htbp]
\centering
\caption{Final Cumulative Control Effort (Radians)}
\label{tab:cumulative_effort}
\begin{tabular}{c c c}
\hline
\textbf{Agent} & \textbf{QP Effort} & \textbf{Auction Effort} \\
\hline
1 & 5.201 & 3.395 \\
2 & 1.610 & 1.664 \\
3 & 1.713 & 2.147 \\
4 & 2.178 & 3.458 \\
\hline
\textbf{Total Effort} & 10.702 & 10.664
\end{tabular}
\end{table}

In the QP case, Agent 1 accumulates a significantly higher total control effort (5.201 rad) compared to the other agents, as it is penalized in every collision. In the auction case, Agent 1's cumulative effort is reduced by approximately 35\% to 3.395 rad. This reduction is achieved because Agent 1's valuation increases with each collision, making it bid more aggressively for avoidance credit in later collisions. This effectively shifts the maneuvering burden onto Agents 2, 3, and 4, thereby balancing the cumulative control effort across the group compared to the baseline. In both cases, the total amount of effort contributed from all agents was nearly identical. 
\section{Conclusion} \label{sec:conclusion}
In this work, we presented a framework for allocating safety responsibility in multi-agent systems. By mapping the HOCBF safety constraint to a divisible resource (avoidance credit), we enabled the use of a progressive second price auction to distribute corrective control effort based on agents' private valuations, without requiring them to reveal their full valuation functions. In this way, the method provides safety guarantees while respecting agent preferences. Algorithm performance was demonstrated in a preliminary experiment in the Robotarium, illustrating how the dynamic valuation scheme enables the system to dynamically adjust the distribution of the corrective burden among agents based on their evolving priorities.

Future work will focus on incorporating feasibility guarantees within control limits and addressing temporal uncertainty via prediction horizons. To improve scalability, we will investigate partitioning agents into decoupled subgroups for parallel auctions. Finally, we plan to explore frameworks to add guarantees for long-term fairness to the system, such as karma economies \cite{elokda_self-contained_2024, berriaud_learning_2023} or external regulators. 

% Original future work paragraph - kinda long

% Future work will focus on incorporating feasibility guarantees to ensure that the safety deficit can always be satisfied within the control limits of agents involved in the auction. We also aim to address the temporal uncertainty of the current bidding process by incorporating prediction horizons for the safety deficit $S(\mathbf{x}, t)$, allowing agents to estimate the cumulative cost of a maneuver before submitting a bid. To improve scalability, we will investigate methods to partition the active agent set into decoupled subgroups, enabling independent conflicts to be resolved via parallel auctions. Finally, while our current approach respects individual preferences, it does not inherently guarantee long-term fairness. To address this, we plan to investigate formal mechanisms, such as karma economies \cite{elokda_self-contained_2024, berriaud_learning_2023} or external regulators, to balance individual agent preferences with communal fairness metrics over extended time frames.
\bibliographystyle{IEEEtran}
\bibliography{reference}
\end{document}